%% file: intMult.tex
\documentclass[11pt]{article}
\usepackage{fullpage}
\input{preamble}

\renewcommand{\R}{\mathcal{R}} 
\renewcommand{\F}{\mathcal{F}} 

\newcommand{\mr}{\mathcal{M}_\mathcal{R}} 

\title{Fast Integer Multiplication Using Modular Arithmetic}
\author{Anindya De, Piyush P Kurur%
  \thanks{Research supported through Research I Foundation project
    NRNM/CS/20030163}, %
  Chandan Saha\\%
  Dept. of Computer Science and Engineering\\%
  Indian Institute of Technology, Kanpur\\%
  Kanpur, UP, India, 208016\\%
  {\tt \{anindya,ppk,csaha\}@cse.iitk.ac.in}%
  \and  Ramprasad Saptharishi%
  \thanks{Research done while visiting IIT Kanpur %
    under Project FLW/DST/CS/20060225}\\%
  Chennai Mathematical Institute\\%
  Plot H1, SIPCOT IT Park\\%
  Padur PO, Siruseri, India, 603103\\%
  and\\%
  Dept. of Computer Science and Engineering\\%
  Indian Institute of Technology, Kanpur\\%
  Kanpur, UP, India, 208016\\%
  {\tt ramprasad@cmi.ac.in}
}
\begin{document}
\maketitle
\begin{abstract}
We give an $O(N\cdot \log N\cdot 2^{O(\log^*N)})$ algorithm for
multiplying two $N$-bit integers that improves the $O(N\cdot \log
N\cdot \log\log N)$ algorithm by
Sch\"{o}nhage-Strassen~\cite{scho}. Both these algorithms use modular
arithmetic. Recently, F\"{u}rer~\cite{Furer} gave an $O(N\cdot \log
N\cdot 2^{O(\log^*N)})$ algorithm which however uses arithmetic over
complex numbers as opposed to modular arithmetic. In this paper, we
use multivariate polynomial multiplication along with ideas {}from
F\"{u}rer's algorithm to achieve this improvement in the modular
setting.  Our algorithm can also be viewed as a $p$-adic version of
F\"{u}rer's algorithm. Thus, we show that the two seemingly different
approaches to integer multiplication, modular and complex arithmetic,
are similar.
\end{abstract}

\section{Introduction}

Computing the product of two $N$-bit integers is an important problem
in algorithmic number theory and algebra. A naive approach leads to an
algorithm that uses $O(N^2)$ bit operations. Karatsuba \cite{Karatsuba} showed that some multiplication operations of
such an algorithm can be replaced by less costly addition operations
which reduces the overall running time of the algorithm to
$O(N^{\log_23})$ bit operations. Shortly afterwards this result was
improved by Toom \cite{Toom} who showed that for any $\varepsilon>0$,
integer multiplication can be done in $O(N^{1+\varepsilon})$
time. This led to the question as to whether the time complexity can
be improved further by replacing the term $O(N^{\epsilon})$ by a
poly-logarithmic factor. In a major breakthrough, Sch\"{o}nhage and
Strassen~\cite{scho} gave two efficient algorithms for multiplying
integers using fast polynomial multiplication. One of the algorithms
achieved a running time of $O(N\cdot \log N\cdot \log\log N\ldots
2^{O(\log^*N)})$ using arithmetic over complex numbers (approximated
to suitable precision), while the other used arithmetic modulo
carefully chosen integers to improve the complexity further to
$O(N\cdot \log N\cdot \log\log N)$. Despite many efforts, the modular
algorithm remained the best until a recent remarkable result by
F\"{u}rer \cite{Furer}. F\"{u}rer gave an algorithm that uses
arithmetic over complex numbers and runs in $O(N\cdot\log N\cdot
2^{O(\log^\ast N)})$ time. Till date this is the best time complexity
result known for integer multiplication.

Sch\"{o}nhage and Strassen introduced two seemingly different
approaches to integer multiplication -- using complex and modular
arithmetic. F\"{u}rer's algorithm improves the time complexity in the
complex arithmetic setting by cleverly reducing some costly
multiplications to simple shift. However, the algorithm
needs to approximate the complex numbers to certain precisions during
computation. This introduces the added task of bounding the total
truncation errors in the analysis of the algorithm. On the contrary,
in the modular setting the error analysis is virtually absent. In
addition, modular arithmetic gives a discrete approach to a discrete
problem like integer multiplication. Therefore it is natural to ask
whether we can achieve a similar improvement in time complexity of
this problem in the modular arithmetic setting. In this paper, we
answer this question affirmatively. We give an $O(N\cdot \log{N}\cdot
2^{O(\log^{*}{N})})$ algorithm for integer multiplication using
modular arithmetic, thus matching the improvement made by F\"{u}rer.

\subsection*{Overview of our result}

As is the case in both Sch\"{o}nhage-Strassen's and F\"{u}rer's
algorithms, we start by reducing the problem to polynomial
multiplication over a ring $\R$ by properly encoding the given
integers. Polynomials can be multiplied efficiently using Discrete
Fourier Transforms (DFT), which uses special roots of unity. For
instance, to multiply two polynomials of degree less than $M$ using
the Fourier transform, we require a \emph{principal} $2M$-th root of
unity (see Definition~\ref{def-principal-root} for \emph{principal
  root}). An efficient way of computing the DFT of a polynomial is
through the Fast Fourier Transform (FFT). In addition, if
multiplications by these roots are efficient, we get a faster
algorithm. Since multiplication by $2$ is a shift, it would be good to
have a ring with $2$ as a root of unity. One way to construct such a
ring in the modular setting is to consider rings of the form $\R =
\Z/(2^M + 1) \Z$ as is the case in Sch\"{o}nhage and
Strassen~\cite{scho}. However, this makes the size of $\R$ equal to
$2^M$, which although works in case of Sch\"{o}nhage and Strassen's
algorithm, is a little too large to handle in our case. We would
like to find a ring whose size is bounded by some polynomial in $M$
and which also contains a principal $2M$-th root of unity. In fact, it
is this choice of ring that poses the primary challenge in adapting
F\"{u}rer's algorithm and making it work in the discrete setting. In
order to overcome this hurdle we choose the ring to be $\R =
\Z/p^c\Z$, for a prime $p$ and a constant $c$ such that $p^c =
poly(M)$. The ring $\Z/p^c\Z$, has a principal $2M$-th root of unity
if and only if $2M$ divides $p-1$, which means that we need to search
for a prime $p$ in the arithmetic progression $\inbrace{1 + i\cdot
  2M}_{i>0}$. To make this search computationally efficient, we need
the degree of the polynomials $M$ to be sufficiently small compared to
the input size. It turns out that this can be achieved by considering
multivariate polynomials instead of univariate polynomials. We use
enough variables to make sure that the search for such a prime does
not affect the overall running time; the number of variables finally
chosen is a constant as well. In fact, the use of multivariate
polynomial multiplications and a small ring are the main steps where
our algorithm differs {}from earlier algorithms by
Sch\"{o}nhage-Strassen and F\"{u}rer.\\

The use of \emph{inner} and \emph{outer} DFT plays a central role in
both F\"{u}rer's as well as our algorithm. Towards understanding the
notion of inner and outer DFT in the context of multivariate
polynomials, we present a group theoretic interpretation of Discrete
Fourier Transform (DFT). Arguing along the line of F\"{u}rer
\cite{Furer} we show that repeated use of efficient computation of
inner DFT's using some special roots of unity in $\R$ makes the
overall process efficient and leads to an $O(N\cdot \log{N}\cdot
2^{O(\log^{*}{N})})$ time algorithm.

\section{The Ring, the Prime and the Root of Unity}\label{ring_section}

We work with the ring $\R = \Z[\alpha]/(p^c, \alpha^m + 1)$ for some
$m$, a constant $c$ and a prime $p$. Elements of $\R$ are thus $m-1$
degree polynomials over $\alpha$ with coefficients {}from $\Z/p^c\Z$. By
construction, $\alpha$ is a $2m$-th root of unity and multiplication
of any element in $\R$ by any power of $\alpha$ can be achieved by
shifting operations --- this property is crucial in making some
multiplications in the FFT less costly
(Section~\ref{fourier_analysis}).

Given an $N$-bit number $a$, we encode it as a $k$-variate polynomial
over $\R$ with degree in each variable less than $M$. The parameters
$M$ and $m$ are powers of two such that $M^k$ is roughly
$\frac{N}{\log^2N}$ and $m$ is roughly $\log{N}$. The parameter $k$
will ultimately be chosen a constant (see Section
\ref{complexity_section}). We now explain the details of this
encoding.

\subsection{Encoding Integers into multivariate
  Polynomials}\label{encoding_section}

Given an $N$-bit integer $a$, we first break these $N$ bits into $M^k$
blocks of roughly $\frac{N}{M^k}$ bits each. This corresponds to
representing $a$ in base $q = 2^{\frac{N}{M^k}}$.  Let $a = a_0 +
\ldots + a_{M^k-1}q^{M^k - 1}$ where $a_i < q$. The number $a$ is
converted into a polynomial as follows:
\begin{enumerate}
\item Express $i$ in base $M$ as $i = i_1 + i_2M + \cdots +
  i_kM^{k-1}$. \label{base_M_item}
\item Encode each term $a_iq^i$ as the monomial $a_i\cdot
  X_1^{i_1}\cdots X_k^{i_k}$. As a result, the number $a$
  gets converted to the polynomial $\sum_{i=0}^{M^k - 1}
  a_i\cdot X_1^{i_1}\cdots X_k^{i_k}$.
\end{enumerate}

Further, we break each $a_i$ into $\frac{m}{2}$ equal sized blocks
where the number of bits in each block is $u = \frac{2N}{M^k\cdot m}$.
Each coefficient $a_i$ is then encoded as polynomial in $\alpha$ of
degree less than $\frac{m}{2}$. The polynomials are then padded with
zeroes to stretch their degrees to $m$. Thus, the $N$-bit number $a$
is converted to a $k$-variate polynomial $a(X)$ over
$\Z[\alpha]/(\alpha^m + 1)$.\\

Given integers $a$ and $b$, each of $N$ bits, we encode them as
polynomials $a(X)$ and $b(X)$ and compute the product polynomial. The
product $a\cdot b$ can be recovered by substituting $X_s =
q^{M^{s-1}}$, for $1\leq s\leq k$, and $\alpha = 2^u$ in the
polynomial $a(X)\cdot b(X)$.  The coefficients in the product
polynomial could be as large as $M^k\cdot m\cdot 2^{2u}$ and hence it
is sufficient to do arithmetic modulo $p^c$ where $p^c > M^k\cdot
m\cdot 2^{2u}$. Therefore, $a(X)$ can indeed be considered as a
polynomial over $\R = \Z[\alpha]/(p^c, \alpha^m + 1)$. Our choice of
the prime $p$ ensures that $c$ is in fact a constant (see
Section~\ref{complexity_section}).

\subsection{Choosing the prime}\label{prime_section}

The prime $p$ should be chosen such that the ring $\Z/p^c\Z$ has a
\emph{principal} $2M$-th root of unity, which is required for
polynomial multiplication using FFT. A principal root of unity is
defined as follows.

\begin{definition}\label{def-principal-root}
An $n$-th root of unity $\zeta\in \R$ is said to be \emph{primitive}
if it generates a cyclic group of order $n$ under
multiplication. Furthermore, it is said to be \emph{principal} if $n$
is coprime to the characteristic of $\R$ and $\zeta$ satisfies
$\sum_{i=0}^{n-1}\zeta^{ij}=0$ for all $0< j < n$.
\end{definition}
\noindent
In $\Z/p^c\Z$, a $2M$-th root of unity is principal if and only if
$2M\mid p-1$ (see also Section~\ref{Qp_section}). As a result, we need
to choose the prime $p$ {}from the arithmetic progression $\inbrace{1 +
i\cdot 2M}_{i> 0}$, which is the main bottleneck of our approach.  We
now explain how this bottleneck can be circumvented.

An upper bound for the least prime in an arithmetic progression is
given by the following theorem \cite{Linnik}:

\begin{theorem}[Linnik]\label{linnik_theorem}
There exist absolute constants $\ell$ and $L$ such that for any pair
of coprime integers $d$ and $n$, the least prime $p$ such that
$p\equiv d\bmod{n}$ is less than $\ell n^L$.
\end{theorem}

Heath-Brown\cite{Brown} showed that the \emph{Linnik constant} $L\leq
5.5$. Recall that $M$ is chosen such that $M^k$ is
$\Theta\inparen{\frac{N}{\log^2N}}$. If we choose $k=1$, that is if we
use univariate polynomials to encode integers, then the parameter $M =
\Theta\inparen{\frac{N}{\log^2N}}$. Hence the least prime $p\equiv
1\pmod{2M}$ could be as large as $N^L$. Since all known deterministic
sieving procedures take at least $N^L$ time this is clearly infeasible
(for a randomized approach see Section~\ref{ERH_section}). However, by
choosing a larger $k$ we can ensure that the least prime $p\equiv
1\pmod{2M}$ is $O(N^\varepsilon)$ for some constant $\varepsilon < 1$.

\begin{remark}\label{prime_time}
If $k$ is any integer greater than $L+1$, then $M^L =
O\inparen{N^{\frac{L}{L+1}}}$ and hence the least prime $p\equiv
1\bmod{2M}$ can be found in $o(N)$ time.
\end{remark}

\subsection{The Root of Unity}\label{root_section}

We require a principal $2M$-th root of unity in $\R$ to compute the
Fourier transforms. This root $\rho(\alpha)$ should also have the
property that its $\inparen{\frac{M}{m}}$-th power is $\alpha$, so as
to make some multiplications in the FFT efficient
(Lemma~\ref{lem-Fmk}). Such a root can be computed by interpolation in
a way similar to that in F\"{u}rer's algorithm\cite[Section 3]{Furer},
but we briefly sketch the procedure for completeness.

We first obtain a $(p-1)$-th root of unity $\zeta$ in $\Z/p^c\Z$ by
lifting a generator of $\mathbb{F}_p^*$. The
$\inparen{\frac{p-1}{2M}}$-th power of $\zeta$ gives us a $2M$-th root
of unity $\omega$. A generator of $\mathbb{F}_p^*$ can be computed by
brute force, as $p$ is sufficiently small. Having obtained a
generator, we can use Hensel Lifting \cite[Theorem 2.23]{Zuckerman}.

\begin{lemma}
Let $\zeta_s$ be a primitive $(p-1)$-th root of unity in
$\Z/p^s\Z$. Then there exists a unique primitive $(p-1)$-th root of
unity $\zeta_{s+1}$ in $\Z/p^{s+1}\Z$ such that $\zeta_{s+1}\equiv
\zeta_s\pmod{p^s}$. This unique root is given by $\zeta_{s+1} =
\zeta_s - \frac{f(\zeta_s)}{f'(\zeta_s)}$ where $f(X) = X^{p-1} - 1$.
\end{lemma}

We need the following claims to compute the root $\rho(\alpha)$.

\begin{claim}\label{claim_for_rho} Let $\omega$ be a principal $2M$-th
  root of unity in
  $\Z/p^c\Z$.
  \begin{enumerate}
  \item[(a)] If $\sigma = \omega^{\frac{M}{m}}$, then $\sigma$ is a
    principal $2m$-th root of unity.
  \item[(b)] The polynomial $x^m+1 = \prod_{i=1}^{m}(x -
    \sigma^{2i-1})$ in $\Z/p^c\Z$. Moreover, for any $0\leq i<j\leq
    2m$, the ideals generated by $(x - \sigma^i)$ and $(x - \sigma^j)$
    are comaximal in $\Z[x]/p^c\Z$.
  \item[(c)] The roots $\inbrace{\sigma^{2i-1}}_{1\leq i \leq m}$ are
    distinct modulo $p$ and therefore the difference of any two of
    them is a unit in $\Z[x]/p^c\Z$.
  \end{enumerate}
\end{claim}

We then, through interpolation, solve for a polynomial $\rho(\alpha)$
such that $\rho(\sigma^{2i+1}) = \omega^{2i+1}$ for all $1\leq i \leq
m$. Then,
\begin{eqnarray*}
  \rho(\sigma^{2i+1}) & = & \omega^{2i+1}\quad 1\leq i \leq m\\
  \implies   \inparen{\rho(\sigma^{2i+1})}^{M/m} & = &
  \omega^{(2i+1)M/m} = \sigma^{2i+1}\\
  \implies \inparen{\rho(\alpha)}^{M/m} & = & \alpha \pmod{\alpha -
    \sigma^{2i+1}}\quad 1\leq i \leq m\\
  \implies \inparen{\rho(\alpha)}^{M/m} & = & \alpha \pmod{ \alpha^m +
    1}
\end{eqnarray*}
The first two parts of the claim justify the Chinese
Remaindering. Finally, computing a polynomial $\rho(\alpha)$ such that
$\rho(\sigma^{2i+1}) = \omega^{2i+1}$ can be done through
interpolation.
$$
\rho(\alpha) = \sum_{i=1}^{m}\omega^{2i+1}\frac{\prod_{j\neq i}(\alpha
  - \sigma^{2j+1})}{\prod_{j\neq i}(\sigma^{2i+1} - \sigma^{2j+1})}
$$ The division by $(\sigma^{2i+1} - \sigma^{2j+1})$ is justified as
it is a unit in $\Z/p^c\Z$ (part (c) of Claim~\ref{claim_for_rho}).

\section{The Integer Multiplication Algorithm}\label{intmult_section}

We are given two integers $a,b< 2^N$ to multiply. We fix constants $k$
and $c$  whose values are given in
Section~\ref{complexity_section}. The algorithm is as follows:
\begin{enumerate}
\item Choose $M$ and $m$ as powers of two such that $M^k \approx
  \frac{N}{\log^2N}$ and $m \approx \log N$. Find the
  least prime $p\equiv 1\pmod{2M}$ (Remark~\ref{prime_time}).
\item Encode the integers $a$ and $b$ as $k$-variate polynomials
  $a(X)$ and $b(X)$ respectively over the ring $\R = \Z[\alpha]/(p^c,
  \alpha^m + 1)$ (Section~\ref{encoding_section}).
\item Compute the root $\rho(\alpha)$ (Section~\ref{root_section}).
\item Use $\rho(\alpha)$ as the principal $2M$-th root of unity to
  compute the Fourier transforms of the $k$-variate polynomials $a(X)$
  and $b(X)$. Multiply component-wise and take the inverse Fourier
  transform to obtain the product polynomial.
\item Evaluate the product polynomial at appropriate powers of two to
  recover the integer product and return it
  (Section~\ref{encoding_section}).
\end{enumerate}

The only missing piece is the Fourier transforms for multivariate
polynomials. The following section gives a group theoretic description
of FFT.

\section{Fourier Transform}\label{sect-FFT}

A convenient way to study polynomial multiplication is to interpret it
as multiplication in a \emph{group algebra}.

\begin{definition}[Group Algebra]
  Let $G$ be a group. The \emph{group algebra} of $G$ over a ring $R$
  is the set of formal sums $\sum_{g \in G} \alpha_g g$ where
  $\alpha_g \in R$ with addition defined point-wise and multiplication
  defined via convolution as follows
  $$ \left(\sum_g \alpha_g g\right) \left(\sum_h
  \beta_h h\right) = \sum_{u}\inparen{\sum_{gh=u} \alpha_g \beta_h}u $$
\end{definition}

Multiplying univariate polynomials over $R$ of degree less than $n$
can be seen as multiplication in the group algebra $R[G]$ where $G$ is
the cyclic group of order $2n$. Similarly, multiplying $k$-variate
polynomials of degree less than $n$ in each variable can be seen as
multiplying in the group algebra $R[G^k]$, where $G^k$ denotes the
$k$-fold product group $G\times\ldots \times G$.

In this section, we study the Fourier transform over the group algebra
$R[E]$ where $E$ is an \emph{additive abelian group}. Most of this,
albeit in a different form, is well known but is provided here for
completeness.\cite[Chapter 17]{Igor}

In order to simplify our presentation, we will fix the base ring to be
$\C$, the field of complex numbers. Let $n$ be the \emph{exponent} of
$E$, that is the maximum order of any element in $E$. A similar
approach can be followed for any other base ring as long as it has a
principal $n$-th root of unity.

We consider $\C[E]$ as a vector space with basis $\{ x \}_{x \in E}$
and use the Dirac notation to represent elements of $\C[E]$ --- the
vector $\ket{x}$, $x$ in $E$, denotes the element $1 . x$ of $\C[E]$.

\begin{definition}[Characters]
Let $E$ be an additive abelian group. A \emph{character} of $E$ is a
homomorphism {}from $E$ to $\C^*$.
\end{definition}

An example of a character of $E$ is the trivial character, which we
will denote by $1$, that assigns to every element of $E$ the complex
number $1$. If $\chi_1$ and $\chi_2$ are two characters of $E$ then
their product $\chi_1 . \chi_2$ is defined as $\chi_1 . \chi_2(x) =
\chi_1(x) \chi_2(x)$.

\begin{proposition}\cite[Chapter 17, Theorem 1]{Igor}
  Let $E$ be an additive abelian group of exponent $n$. Then the
  values taken by any character of $E$ are $n$-th roots of
  unity. Furthermore, the characters form a \emph{multiplicative
    abelian group} $\hat{E}$ which is isomorphic to $E$.
\end{proposition}

An important property that the characters satisfy is the following
\cite[Corollary 2.14]{Isaacs}.

\begin{proposition}[Schur's Orthogonality]
  Let $E$ be an additive abelian group. Then
  \[ \sum_{x \in E} \chi(x) =%
  \begin{cases}
    0 & \textrm{ if $\chi \neq 1$,}\\
    \# E &\textrm{ otherwise}
  \end{cases}
  \]
  \[
  \sum_{\chi \in \hat{E}} \chi(x) =%
  \begin{cases}
    0 & \textrm{ if $x \neq 0$,}\\
    \# E &\textrm{ otherwise.}
  \end{cases}
  \]
\end{proposition}

It follows {}from Schur's orthogonality that the collection of vectors
$\ket{\chi} = \sum_x \chi(x) \ket{x}$ forms a basis of $\C[E]$. We
will call this basis the \emph{Fourier basis} of $\C[E]$.

\begin{definition}[Fourier Transform]
Let $E$ be an additive abelian group and let $x \mapsto \chi_x$ be an
isomorphism between $E$ and $\hat{E}$. The \emph{Fourier transform}
over $E$ is the linear map {}from $\C[E]$ to $\C[E]$ that sends
$\ket{x}$ to $\ket{\chi_x}$.
\end{definition}

Thus, the Fourier transform is a change of basis {}from the point basis
$\{ \ket{x} \}_{x \in E}$ to the Fourier basis $\{
\ket{\chi_x}\}_{x\in E}$. The Fourier transform is unique only up to
the choice of the isomorphism $x \mapsto \chi_x$. This isomorphism is
determined by the choice of the principal root of unity.

\begin{remark}\label{rem-Fourier-inner}
Given an element $\ket{f} \in \C[E]$, to compute its Fourier transform
it is sufficient to compute the \emph{Fourier coefficients}
$\{\braket{\chi}{f} \}_{\chi \in \hat{E}}$.
\end{remark}

\subsection{Fast Fourier Transform}

We now describe the Fast Fourier Transform for general abelian groups
in the character theoretic setting. For the rest of the section fix an
additive abelian group $E$ over which we would like to compute the
Fourier transform. Let $A$ be any subgroup of $E$ and let $B =
E/A$. For any such pair of abelian groups $A$ and $B$, we have an
appropriate Fast Fourier transformation, which we describe in the rest
of the section.

\begin{proposition}\label{prop-character-lift}
  \begin{enumerate}
  \item Every character $\lambda$ of $B$ can be ``lifted'' to a
    character of $E$ (which will also be denoted by $\lambda$) defined
    as follows $\lambda(x) = \lambda(x + A)$.
  \item Let $\chi_1$ and $\chi_2$ be two characters of $E$ that when
    restricted to $A$ are identical. Then $\chi_1 = \chi_2 \lambda$ for
    some character $\lambda$ of $B$.
  \item The group $\hat{B}$ is (isomorphic to) a subgroup of $\hat{E}$
    with the quotient group $\hat{E}/\hat{B}$ being (isomorphic to)
    $\hat{A}$.
  \end{enumerate}
\end{proposition}

We now consider the task of computing the Fourier transform of an
element $\ket{f} = \sum f_x \ket{x}$ presented as a list of
coefficients $\{f_x\}$ in the point basis. For this, it is sufficient
to compute the Fourier coefficients $\{\braket{\chi}{f}\}$ for each
character $\chi$ of $E$ (Remark~\ref{rem-Fourier-inner}). To describe
the Fast Fourier transform we fix two sets of cosets representatives,
one of $A$ in $E$ and one of $\hat{B}$ in $\hat{E}$ as follows.

\begin{enumerate}
  \item For each $b \in B$, $b$ being a coset of $A$, fix a coset
    representative $x_b \in E$ such $b = x_b + A$.
  \item For each character $\varphi$ of $A$, fix a character
    $\chi_\varphi$ of $E$ such that $\chi_\varphi$ restricted to $A$ is
    the character $\varphi$. The characters $\{ \chi_\varphi \}$ form
    (can be thought of as) a set of coset representatives of $\hat{B}$
    in $\hat{E}$.
\end{enumerate}

Since $\{ x_b \}_{b \in B}$ forms a set of coset representatives, any
$\ket{f} \in \C[E]$ can be written uniquely as $\ket{f} = \sum f_{b,a}
\ket{x_b + a}$.

\begin{proposition}\label{prop-Fourier-coefficient}
  Let $\ket{f} = \sum f_{b,a}\ket{x_b + a}$ be an element of $\C[E]$.
  For each $b \in B$ and $\varphi \in \hat{A}$ let $\ket{f_b}\in
  \C[A]$ and $\ket{f_\varphi} \in \C[B]$ be defined as
  follows.
  \begin{eqnarray*}
    \ket{f_b} &= & \sum_{a \in A} f_{b,a} \ket {a}\\
    \ket{f_\varphi} & = &\sum_{b \in B} \overline{\chi}_{\varphi}(x_b)
    \braket{\varphi}{f_b} \ket{b}
  \end{eqnarray*}
  Then for any character $\chi$ of $E$, which can be expressed as
  $\chi = \lambda\cdot\chi_\varphi$, the Fourier coefficient
  $\braket{\chi}{f} = \braket{\lambda}{f_\varphi}$.
\end{proposition}
\begin{proof}Recall that $\lambda(x+A) = \lambda(x)$, and $\varphi$ is
  a restriction of the $\chi$ to the subgroup $A$.
\begin{eqnarray*}
\braket{\chi}{f} & = &
\sum_b\sum_a\overline{\chi}(x_b+a)f_{b,a}\\
& = & \sum_b\overline{\lambda}(x_b+a)\sum_a\overline{\chi}_\varphi(x_b+a)f_{b,a}\\
& = &
\sum_b\overline{\lambda}(b)\overline{\chi}_\varphi(x_b)\sum_a\overline{\varphi}(a)f_{b,a}\\
& = &
\sum_b\overline{\lambda}(b)\overline{\chi}_\varphi(x_b)\braket{\varphi}{f_b}\\
& = & \braket{\lambda}{f_\varphi}
\end{eqnarray*}
\end{proof}
We are now ready to describe the Fast Fourier transform given an
element $\ket{f} = \sum f_x \ket{x}$.
\begin{enumerate}
\item \label{step_inner_dft} For each $b \in B$ compute the Fourier
  transforms of $\ket{f_b}$. This requires $\# B$ many Fourier
  transforms over $A$.
\item \label{step_bad_mult}As a result of the previous step we have
  for each $b \in B$ and $\varphi \in \hat{A}$ the Fourier
  coefficients $\braket{\varphi}{f_b}$. Compute for each $\varphi$ the
  vectors $\ket{f_\varphi} = \sum_{b \in B}
  \overline{\chi}_{\varphi}(x_b) \braket{\varphi}{f_b} \ket{b}$. This
  requires $\# \hat{A} . \# B = \# E$ many multiplications by roots of
  unity.
\item \label{step_outer_dft} For each $\varphi \in \hat{A}$ compute
  the Fourier transform of $\ket{f_\varphi}$. This requires $\#\hat{A}
  = \# A$ many Fourier transforms over $B$.\label{item-Fourier-B}
\item Any character $\chi$ of $E$ is of the
  form $\chi_\varphi \lambda$ for some $\varphi \in \hat{A}$ and
  $\lambda \in \hat{B}$. Using
  Proposition~\ref{prop-Fourier-coefficient} we have at the end of
  Step~\ref{item-Fourier-B} all the Fourier coefficients
  $\braket{\chi}{f} = \braket{\lambda}{f_\varphi}$.
\end{enumerate}

If the quotient group $B$ itself has a subgroup that is isomorphic to
$A$ then we can apply this process recursively on $B$ to obtain a
divide and conquer procedure to compute the Fourier transform. In the
standard FFT we use $E = \Z/2^n\Z$. The subgroup $A$ is $2^{n-1}E$
which is isomorphic to $\Z/2\Z$ and the quotient group $B$ is
$\Z/2^{n-1}\Z$.

\subsection{Analysis of the Fourier Transform}\label{fourier_analysis}

Our goal is to multiply $k$-variate polynomials over $\R$, with the
degree in each variable less than $M$. This can be achieved by
embedding the polynomials into the algebra of the product group $E =
\inparen{\frac{\Z}{2M\cdot \Z}}^k$ and multiplying them as elements of
the algebra. Since the exponent of $E$ is $2M$, we require a principal
$2M$-th root of unity in the ring $\R$. We shall use the root
$\rho(\alpha)$ (as defined in Section~\ref{root_section}) for the
Fourier transform over $E$.

For every subgroup $A$ of $E$, we have a corresponding FFT. We choose
the subgroup $A$ as $\inparen{\frac{\Z}{2m\cdot \Z}}^k$ and let $B$ be
the quotient group $E/A$. The group $A$ has exponent $2m$ and $\alpha$
is a principal $2m$-th root of unity. Since $\alpha$ is a power of
$\rho(\alpha)$, we can use it for the Fourier transform over $A$. As
multiplications by powers of $\alpha$ are just shifts, this makes
Fourier transform over $A$ efficient.

Let $\F(2M,k)$ denote the complexity of computing the Fourier transform
over $\inparen{\frac{\Z}{2M\cdot \Z}}^k$. We have
\begin{equation}
\F(2M,k) = \inparen{\frac{M}{m}}^k \F(2m,k) + M^k
\mathcal{M}_\R+(2m)^k\F\inparen{\frac{M}{m},k}
\label{first_recursive_step}
\end{equation}
where $\mathcal{M}_\R$ denotes the complexity of multiplications in
$\R$. The first term comes {}from the $\# B$ many Fourier transforms
over $A$ (Step~\ref{step_inner_dft} of FFT), the second term
corresponds to the multiplications by roots of unity
(Step~\ref{step_bad_mult}) and the last term comes {}from the $\# A$
many Fourier transforms over $B$ (Step~\ref{step_outer_dft}).

Since $A$ is a subgroup of $B$ as well, Fourier transforms over $B$
can be recursively computed in a similar way, with $B$ playing the
role of $E$. Therefore, by simplifying the recurrence in
Equation~\ref{first_recursive_step} we get:
\begin{equation}
\F(2M,k) = O\inparen{\frac{M^k\log M}{m^k\log m}\F(2m,k) +
  \frac{M^k\log M}{\log m}\mr}
\label{eqn_with_Fa}
\end{equation}

\begin{lemma}\label{lem-Fmk}
$\F(2m,k) = O(m^{k+1}\log m\cdot \log p)$
\end{lemma}
\begin{proof}
The FFT over a group of size $n$ is usually done by taking $2$-point
FFT's followed by $\frac{n}{2}$-point FFT's. This involves $O(n\log
n)$ multiplications by roots of unity and additions in base
ring. Using this method, Fourier transforms over $A$ can be computed
with $O(m^k\log m)$ multiplications and additions in $\R$. Since each
multiplication is between an element of $\R$ and a power of $\alpha$,
this can be efficiently achieved through shifting operations. This is
dominated by the addition operation, which takes $O(m\log p)$ time,
since this involves adding $m$ coefficients {}from $\Z/p^c\Z$.
\end{proof}

Therefore, {}from Equation~\ref{eqn_with_Fa},
\begin{equation}\label{eqn_FFT_complexity}
\F(2M,k) = O\inparen{M^k\log M\cdot m\cdot \log p + \frac{M^k\log
    M}{\log m}\mr}
\end{equation}

\medskip

\section{Complexity Analysis}\label{complexity_section}

The choice of parameters should ensure that the following constraints
are satisfied:
\begin{enumerate}
\item $M^k = \Theta\inparen{\frac{N}{\log^2N}}$ and $m = O(\log
  N)$.
\item $M^L = O(N^\varepsilon)$ where $L$ is the Linnik constant
  (Theorem~\ref{linnik_theorem}) and $\varepsilon$ is any constant less
  than $1$. Recall that this makes picking the prime by brute force
  feasible (see Remark~\ref{prime_time}).
\item $p^c > M^k\cdot m\cdot 2^{2u}$ where $u = \frac{2N}{M^km}$. This
  is to prevent overflows during modular arithmetic (see
  Section~\ref{encoding_section}).
\end{enumerate}
\noindent
It is straightforward to check that $k > L+1$ and $c > 5(k+1)$ satisfy
the above constraints. Heath-Brown~\cite{Brown} showed that $L\leq
5.5$ and therefore $c = 42$ clearly suffices.\\

Let $T(N)$ denote the time complexity of multiplying two $N$ bit
integers. This consists of:

\begin{enumerate}
\item[(a)] Time required to pick a suitable prime $p$,
\item[(b)] Computing the root $\rho(\alpha)$,
\item[(c)] Encoding the input integers as polynomials,
\item[(d)] Multiplying the encoded polynomials,
\item[(e)] Evaluating the product polynomial.
\end{enumerate}

As argued before, the prime $p$ can be chosen in $o(N)$ time. To
compute $\rho(\alpha)$, we need to lift a generator of
$\mathbb{F}_p^*$ to $\Z/p^c\Z$ followed by an interpolation. Since $c$
is a constant and $p$ is a prime of $O(\log N)$ bits, the time
required for Hensel Lifting and interpolation is $o(N)$.

The encoding involves dividing bits into smaller blocks, and
expressing the exponents of $q$ in base $M$
(Section~\ref{encoding_section}) and all these take $O(N)$ time since
$M$ is a power of $2$. Similarly, evaluation of the product polynomial
takes linear time as well. Therefore, the time complexity is dominated
by the time taken for polynomial multiplication.

\subsubsection*{Time complexity of Polynomial Multiplication}

{}{}From Equation~\ref{eqn_FFT_complexity}, the complexity of Fourier
transform is given by
\[
\F(2M,k) = O\inparen{M^k\log M\cdot m\cdot \log p + \frac{M^k\log
    M}{\log m}\mr}
\]

\begin{proposition}\cite{scho_complex}
Multiplication in the ring $\R$ reduces to multiplying $O(\log^2{N})$
bit integers and therefore $\mr = T\left(O(\log^2{N})\right)$.
\end{proposition}
\begin{proof}
  Elements of $\R$ can be seen as polynomials in $\alpha$ over
  $\Z/p^c\Z$ with degree at most $m$. Given two such polynomials
  $f(\alpha)$ and $g(\alpha)$, encode them as follows: Replace $\alpha$
  by $2^d$, transforming the polynomials $f(\alpha)$ and $g(\alpha)$
  to the integers $f(2^d)$ and $g(2^d)$ respectively.  The parameter
  $d$ is chosen such that the coefficients of the product $h(\alpha) =
  f(\alpha) g(\alpha)$ can be recovered {}from the product $f(2^d)\cdot
  g(2^d)$. For this, it is sufficient to ensure that the maximum
  coefficient of $h(\alpha)$ is less than $2^d$.  Since $f$ and $g$
  are polynomials of degree $m$, we would want $2^d$ to be greater
  than $m\cdot p^{2c}$, which can be ensured by choosing $d =
  \Theta\left(\log{N}\right)$. The integers $f(2^d)$ and $g(2^d)$ are
  bounded by $2^{md}$ and hence the task of multiplying in $\R$
  reduces to $O(\log^2{N})$ bit integer multiplication.
\end{proof}

Multiplication of two polynomials involve a Fourier transform followed
by component-wise multiplications and an inverse Fourier
transform. Since the number of component-wise multiplications is only
$M^k$, the time taken is $M^k\cdot\mr$ which is clearly subsumed in
$\F(M,k)$. Therefore, the time taken for multiplying the polynomials
is $O(\F(M,k))$. Thus, the complexity of our integer multiplication
algorithm $T(N)$ is given by,
\begin{eqnarray*}
T(N) & = & O(\F(M,k))\\
& = & O\inparen{M^k\log M\cdot m\cdot \log p + \frac{M^k\log
  M}{\log m}\mr}\\
& = & O\inparen{N\log N + \frac{N}{\log N\cdot \log\log N}T(O(\log^2N))}
\end{eqnarray*}

The above recurrence leads to the following theorem.

\begin{theorem}
Given two $N$ bit integers, their product can be computed in
$O(N\cdot \log N\cdot 2^{O(\log^*N)})$ time.
\end{theorem}
\subsection{Choosing the Prime Randomly}\label{ERH_section}

To ensure that the search for a prime $p\equiv 1\pmod{2M}$ does not
affect the overall time complexity of the algorithm, we considered
multivariate polynomials to restrict the value of $M$; an alternative is
to use randomization.

\begin{proposition}
Assuming ERH, a prime $p$ such that $p\equiv 1\pmod{2M}$ can be
computed by a randomized algorithm with expected running time
$\tilde{O}(\log^3 M)$.
\end{proposition}
\begin{proof}
Titchmarsh\cite{Titchmarsh} (see also Tianxin \cite{Tianxin})
showed, assuming ERH, that the number of primes less than $x$ in the
arithmetic progression $\{ 1 + i \cdot 2M\}_{i > 0}$ is given by,
\begin{equation*}
\pi(x,2M) = \frac{Li(x)}{\varphi(2M)} + O(\sqrt{x} \log x)
\end{equation*}
for $2M \leq \sqrt{x} \cdot (\log x)^{-2}$, where $Li(x) =
\Theta(\frac{x}{\log x})$ and $\varphi$ is the Euler totient
function. In our case, since $M$ is a power of two, $\varphi(2M) = M$,
and hence for $x \geq 4M^2 \cdot \log^6 M$, we have $\pi(x, 2M) =
\Omega\inparen{\frac{x}{M\log x}}$. Therefore, for an $i$ chosen
uniformly randomly in the range $1 \leq i \leq 2M \cdot \log^6 M$, the
probability that $i\cdot 2M + 1$ is a prime is at least $\frac{d}{\log
  x}$ for a constant $d$. Furthermore, primality test of an $O(\log
M)$ bit number can be done in $\tilde{O}(\log^2 M)$ time using
Rabin-Miller primality test \cite{Miller, Rabin}. Hence, with $x = 4M^2
\cdot \log^6 M$ a suitable prime for our algorithm can be found in
expected $\tilde{O}(\log^3 M)$ time.
\end{proof}

\section{A Different Perspective } \label{Qp_section}

Our algorithm can be seen as a $p$-adic version of F\"{u}rer's integer
multiplication algorithm, where the field $\C$ is replaced by $\Q_p$,
the field of $p$-adic numbers (for a quick introduction, see Baker's
online notes\cite{Baker}). Much like $\C$, where representing a
general element (say in base $2$) takes infinitely many bits,
representing an element in $\Q_p$ takes infinitely many $p$-adic
digits. Since we cannot work with infinitely many digits, all
arithmetic has to be done with finite precision. Modular arithmetic in
the base ring $\Z[\alpha]/(p^c, \alpha^m + 1)$, can be viewed as
arithmetic in the ring $\Q_p[\alpha]/(\alpha^m + 1)$ keeping a
precision of $\varepsilon = p^{-c}$.

Arithmetic with finite precision naturally introduces some errors in
computation. However, the nature of $\Q_p$ makes the error analysis
simpler. The field $\Q_p$ comes with a norm $\abs{\ \cdot\ }_p$ called
the $p$-adic norm, which satisfies the stronger triangle inequality
$\abs{x+y}_p \leq \max\inparen{\abs{x}_p, \abs{y}_p}$\cite[Proposition
  2.6]{Baker}. As a result, unlike in $\C$, the errors in computation
do not compound.\\

Recall that the efficiency of FFT crucially depends on a special
principal $2M$-th root of unity in $\Q_p[\alpha]/(\alpha^m + 1)$. Such
a root is constructed with the help of a primitive $2M$-th root of
unity in $\Q_p$. The field $\Q_p$ has a primitive $2M$-th root of
unity if and only if $2M$ divides $p-1$\cite[Theorem
  5.12]{Baker}. Also, if $2M$ divides $p-1$, a $2M$-th root can be
obtained {}from a $(p-1)$-th root of unity by taking a suitable power. A
primitive $(p-1)$-th root of unity in $\Q_p$ can be constructed, to
sufficient precision, using Hensel Lifting starting {}from a generator
of $\mathbb{F}_p^*$.

\section{Conclusion}\label{conclusions_section}

There are two approaches for multiplying integers, one using
arithmetic over complex numbers, and the other using modular
arithmetic. Using complex numbers, Sch\"{o}nhage and
Strassen\cite{scho} gave an $O(N \cdot \log N \cdot \log\log N\ldots
2^{O(\log^* N)})$ algorithm. F\"{u}rer\cite{Furer} improved this
complexity to $O(N\cdot\log N \cdot2^{O(\log^*N)})$ using some special
roots of unity. The other approach, that is modular arithmetic, can be
seen as arithmetic in $\Q_p$ with certain precision. A direct
adaptation of the Sch\"{o}nhage-Strassen's algorithm in the modular
setting leads to an $O(N \cdot \log N \cdot \log\log N\ldots
2^{O(\log^* N)})$ algorithm. In this paper, we showed that by choosing
an appropriate prime and a special root of unity, a running time of
$O(N\cdot \log N \cdot 2^{O(\log^*N)})$ can be achieved through
modular arithmetic as well. Therefore, in a way, we have unified the
two paradigms.

\subsection*{Acknowledgement}

We thank V. Vinay, Srikanth Srinivasan and the anonymous referees for
many helpful suggestions that improved the overall presentation of
this paper.

\bibliography{references}

\end{document}

%% file: preamble.tex
\usepackage{latexsym}
\usepackage{amsmath}
\usepackage{amssymb}
\usepackage{amsthm}
\usepackage{graphicx}
\usepackage{hyperref}
\usepackage{cite}

\newtheorem{theorem}{Theorem}[section]
\newtheorem{theorem*}{Theorem}

\newtheorem{lemma}[theorem]{Lemma}

\newtheorem{proposition}[theorem]{Proposition}
\newtheorem{definition}[theorem]{Definition}
\newtheorem{claim}[theorem]{Claim}

\theoremstyle{definition}\newtheorem{remark}[theorem]{Remark}

\newcommand{\inparen}[1]{\left(#1\right)}             
\newcommand{\inbrace}[1]{\left\{#1\right\}}           
\newcommand{\abs}[1]{\left|#1\right|}                 

\newcommand{\ket}[1]{\left| #1 \right\rangle}
\newcommand{\braket}[2]{\langle #1|#2 \rangle}

\newcommand{\F}{\mathbb{F}}

\newcommand{\Q}{\mathbb{Q}}
\newcommand{\Z}{\mathbb{Z}}
\newcommand{\R}{\mathbb{R}}
\newcommand{\C}{\mathbb{C}}